\documentclass{article}
\usepackage[utf8]{inputenc}

\usepackage[square,comma,numbers]{natbib}
\usepackage{color, colortbl}
\usepackage[first=0,last=9]{lcg}

\usepackage{amsmath}
\usepackage[english]{babel}
\usepackage[version=3]{mhchem}
\usepackage{graphicx}
\usepackage{comment}
\usepackage{tikz}

\usepackage{pdfpages}
\definecolor{LightCyan}{rgb}{0.88, 1, 1}
\definecolor{purple}{rgb}{1, .79, 1}
\definecolor{Yellow}{rgb}{1, 1, .6}
\definecolor{Gray}{gray}{0.85}
\definecolor{Almond}{rgb}{0.94, 0.87, 0.8}
\definecolor{antiquebrass}{rgb}{0.8, 0.58, 0.46}
\definecolor{applegreen}{rgb}{0.55, 0.71, 0.0}
\definecolor{babypink}{rgb}{0.96, 0.76, 0.76}
\definecolor{aqua}{rgb}{0.0, 1.0, 1.0}
\usepackage{enumerate}
\usepackage{amsthm}
\usepackage{amsmath, amssymb}

\usepackage[export]{adjustbox}
\usepackage{verbatim}
\usepackage{graphicx}

\newtheorem{theorem}{Theorem}
\usepackage{array}
\usepackage[leftcaption]{sidecap}
\usepackage{pst-node}
\usetikzlibrary{shapes.geometric}%
\usepackage[colorinlistoftodos]{todonotes}
\usepackage[colorlinks=true, allcolors=blue]{hyperref}

\usepackage{setspace}
\usepackage{fullpage}
\usepackage{times}
\usepackage{fancyhdr,graphicx,amsmath,amssymb}
\usepackage[ruled,vlined]{algorithm2e}
\include{pythonlisting}
\usepackage{neuralnetwork}

\makeatletter

\begin{document}

\title{Normalized Augmented Inverse Probability Weighting with Neural Network Predictions}

\author{ 
Mehdi Rostami, Olli Saarela

\\ 
Biostatistics, Dalla Lana School of Public Health, University of Toronto, ON, Canada}

\maketitle

\begin{abstract}

The estimation of Average Treatment Effect (ATE) as a causal parameter is carried out in two steps, where in the first step, the treatment and outcome are modeled to incorporate the potential confounders, and in the second step, the predictions are inserted into the ATE estimators such as the Augmented Inverse Probability Weighting (AIPW) estimator. Due to the concerns regarding the nonlinear or unknown relationships between confounders and the treatment and outcome, there has been an interest in applying non-parametric methods such as Machine Learning (ML) algorithms instead. \citet{farrell2018deep} proposed to use two separate Neural Networks (NNs) where there's no regularization on the network's parameters except the Stochastic Gradient Descent (SGD) in the NN's optimization. Our simulations indicate that the AIPW estimator suffers extensively if no regularization is utilized. We propose the normalization of AIPW (referred to as nAIPW) which can be helpful in some scenarios. nAIPW, provably, has the same properties as AIPW, that is, the double-robustness and orthogonality properties \citep{chernozhukov2018double}. Further, if the first step algorithms converge fast enough, under regulatory conditions \citep{chernozhukov2018double}, nAIPW will be asymptotically normal. We also compare the performance of AIPW and nAIPW in terms of the bias and variance when small to moderate $L_1$ regularization is imposed on the NNs.

\end{abstract}


\section{Introduction}

Estimation of causal parameters such as the Average Treatment Effect (ATE) in observational data requires confounder adjustment. The estimation and inference are carried out in two steps: In step 1, the treatment and outcome are predicted by a statistical models or machine learning (ML) algorithm, and in the second step the predictions are inserted into the causal effect estimator. If ML algorithms are employed in step 1, the non-linear relationships can potentially be taken into account. The relationship between the confounders and the treatment and outcome can be non-linear which make the application of Machine Learning (ML) algorithms which are non-parametric models appealing. \citet{farrell2018deep} proposed to use two separate Neural Networks (double NNs or dNN) where there's no regularization on the network's parameters except the Stochastic Gradient Descent (SGD) in the NN's optimization. They derive the generalization bounds and prove that the NNs algorithms are fast enough so that the asymptotic distribution of causal estimators such as the Augmented Inverse Probability Weighting (AIPW) estimator \citep{lunceford2004stratification} will be asymptotically linear, under regulatory conditions and the utilization of cross-fitting \cite{chernozhukov2018double}. 

\citet{farrell2018deep} argue that the fact that SGD-type algorithms control the complexity of the NN algorithm up to some extent \citep{goodfellow2016deep, zhang2021understanding} is sufficient for the first step. Our initial simulations and analyses, however, contradict this claim in scenarios where strong confounders and instrumental variables exist in the data. We argue that for causal parameter estimation, dNN with no regularization leads to high variance for the causal estimator used in the second step. This is also intuitively true as the complexity of NNs can lead to perfect prediction in the treatment model which violates the positivity assumption \citep{diaz2018doubly}. This is particularly true for the AIPW estimator where the propensity score predictions in the denominator can blow up the estimates (and the standard error estimations.) 

We propose and study a quick remedy to this issue by studying the normalization of the AIPW estimator (similar to the normalization of the Inverse Probability Weighting (IPW) \citep{lunceford2004stratification}), here referred to as nAIPW. In fact, both AIPW and nAIPW can be viewed as a more general estimator which is derived via the efficient influence function of ATE \cite{hines2021demystifying}. 

To the best of our knowledge, the performance of nAIPW has not been previously studied in the machine learning context. We will prove that this estimator has the doubly robust \citep{robins1994estimation} and the rate doubly robust \citep{hines2021demystifying} property, and illustrate that it is robust against extreme propensity score values. Further, nAIPW (similar to AIPW), has the orthogonality property \citep{chernozhukov2018double} which means that it
is robust against small variation in the predictions of the outcome and treatment assignment predictions. One theoretical difference is that AIPW is the most efficient estimator among all the double robust estimators of ATE given both treatment and outcome models are correctly specified \citep{scharfstein1999adjusting}. In practice, however, often there is no a priori knowledge about the true outcome and propensity score relationships with the input covariates and thus this feature of AIPW is probably of less practical use. 

We compare AIPW and nAIPW through a simulation study where we allow for moderate to strong confounding and instrumental variable effects, that is, allow for possible violation of the positivity assumption. Further, a comparison between AIPW and nAIPW is made on the Canadian Community Health Survey (CCHS) dataset where the intervention/treatment is the food security vs food insecurity and the outcome is individuals' Body Mass Index (BMI). 

Organization of the article is as follows. In Section \ref{nAIPWintro} we will formally introduce the nAIPW estimator to the readers and state its double robustness property, and in Section \ref{prediction-model} we present the first step prediction model, double neural networks. In section \ref{theory} we will present the theoretical aspects of the paper, including the asymptotic normality, doubly robustness and rate doubly robustness orthogonality of the proposed estimator (nAIPW) and the asymptotic normality. We will present the simulation scenarios and results of comparing the nAIPW estimator with other conventional estimators in Section \ref{simulations}. We apply the estimators on a real dataset in Section \ref{applications}. The article will be concluded with a short discussion on the findings in section \ref{discussion}. The proofs are straightforward but long and thus are included in the Appendix \ref{appendix}.

\section{Normalized Doubly Robust Estimator}\label{nAIPWintro}

Let data $\mathbf{O}=(O_1, O_2, ..., O_n)$ be generated by a data generating process $P$, where $O_i$ is a finite dimensional vector $O_i=(Y_i, A_i, W_i)$, with $\mathbf{W}$ being the adjusting factors. $P$ is the true observed data distribution, $\hat{P}_n$ is the distribution of $\mathbf{O}$ such that its marginal distribution with respect to $W$ is its empirical distribution, and the expectation of the conditional distribution $Y|A=a, W$, for $a=0,1$, can be estimated. We denote the prediction function of observed outcome given explanatory variables in the treated group $Q^1:=Q(1, W)=\mathbb{E}[Y|A=1, W]$, and that in the untreated group $Q^0:=Q(0, W)=\mathbb{E}[Y|A=0, W]$, and the propensity score as $g(W)=\mathbb{E}[A|W]$. Throughout, the expectations $\mathbb{E}$ are with respect to $P$. The symbol $\ \hat{}\ $ on the population-level quantities indicates the corresponding finite sample estimator, and $P$ is replaced by $\hat{P}_n$.

Let the causal parameter of interest be the Average Treatment Effect (ATE)
\begin{equation}\label{causalriskdiff0}
   \beta_{ATE} = \mathbb{E}[Y^1 - Y^0] = \mathbb{E}\big[\mathbb{E}[Y^1 - Y^0|W]\big]=\mathbb{E}\big[\mathbb{E}[Y|A=1,W]\big]-\mathbb{E}\big[\mathbb{E}[Y|A=0,W]\big],
\end{equation}
where $Y^1$ and $Y^0$ are the potential outcomes of the treatment and controls \citep{robins1994estimation}. 

For identifiablity of the parameter, the following assumptions must hold true. The first assumption is the Conditional Independence, or Unconfoundedness stating that, given the confounders, the potential outcomes are independent of the treatment assignments ($Y^0, Y^1 \perp A| W$). The second assumption is Positivity which entails that the assignment of treatment groups is not deterministic ($0 < Pr(A = 1|W) < 1$). The third assumption is Consistency which states that the observed outcomes equal their corresponding potential outcomes ($Y^{A} = y$). There are other modeling assumptions made such as time order (i.e. the covariates $W$ are measured before the treatment), IID subjects, and a linear causal effect.

A list of first candidates to estimate ATE are

\begin{equation} \label{allests}
\begin{split}
        \text{nATE} & \qquad \hat{\beta}_{nATE} = \frac{1}{n_1}\sum_{i\in A_1} \hat{Q}^1_i - \frac{1}{n_0}\sum_{i\in A_0}\hat{Q}^0_i,\\
        \text{SR}& \qquad \hat{\beta}_{SR} = \hat{\mathbb{E}}\Big[\hat{\mathbb{E}}[Y^1 - Y^0|W]\Big] = \frac{1}{n}\sum_{i=1}^n \hat{Q}^1_i - \hat{Q}^0_i,\\ 
        \text{IPW}& \qquad \beta_{IPW} = \hat{\mathbb{E}}\Big[\frac{Y^1}{\hat{\mathbb{E}}[A|W]}-\frac{Y^0}{1-\hat{\mathbb{E}}[A|W]}\Big] = \frac{1}{n}\sum_{i=1}^n\Big(\frac{A_iy_i}{\hat{g}_i}-\frac{(1-A_i)y_i}{1-\hat{g}_i}\Big),\\ 
        \text{nIPW}& \qquad \hat{\beta}_{nIPW} = \sum_{i=1}^n \Big( \frac{A_iw^{(1)}_iy_i}{\sum_{j=1}^n A_jw^{(1)}_j}-\frac{(1-A_i)w^{(0)}_iy_i}{\sum_{j=1}^n (1-A_j)w^{(0)}_j}\Big).
\end{split}
\end{equation}

The naive Average Treatment Effect (nATE) is a biased (due to the selection bias) estimator of ATE \citep{angrist2008mostly} and is poorest estimator among all the candidates. The Single Robust (SR) is not an orthogonal estimator \citep{chernozhukov2018double} and if ML algorithms which do not belong to the Donsker class (\cite{van2000asymptotic}, Section 19.2) or their entropy grows with the sample size are used, this estimator also becomes biased and is not asymptotically normal. The Inverse Probability Weighting (IPW) \citep{horvitz1952generalization} and its normalization versions adjust (or weight) the observations in the treatment and control groups. IPW and nIPW are also not orthogonal estimators and are similar to SR in this respect. In addition, both $\hat{\beta}_{SR}$ and $\hat{\beta}_{IPW}$ (and $\hat{\beta}_{nIPW}$) are single robust, that is, they are consistent estimators of ATE if the models used are $\sqrt{n}$-consistent \citep{lunceford2004stratification}. IPW is an unbiased estimator of ATE if $g$ is correctly specified, but nIPW is not unbiased, but is less sensitive to extreme predictions. The Augmented Inverse Probability Weighting (AIPW) estimator \citep{scharfstein1999adjusting} is an improvement over SR, IPW and nIPW which involves the predictions for both treatment (the propensity score) and the the causal parameter can be expressed as:

\begin{equation}\label{trueAIPW}
   \beta = \mathbb{E}\Bigg[\Big(\frac{AY-Q(1,W)(A-\mathbb{E}[A |W])}{\mathbb{E}[A |W]} \Big)
   -
   \Big(\frac{(1-A)Y+Q(0,W)(A-\mathbb{E}[A |W])}{1-\mathbb{E}[A |W]}\Big)\Bigg],
\end{equation}
and the sample version estimator of \eqref{trueAIPW} is
\begin{multline}\label{AIPW}
    \hat{\beta}_{AIPW} = \frac{1}{n}\sum_{i=1}^n\Bigg[\Big(\frac{A_iY_i-\hat{Q}(1, W_i)(A_i-\hat{\mathbb{E}}[A_i |W_i])}{\hat{\mathbb{E}}[A_i |W_i]} \Big)
   -
   \Big(\frac{(1-A_i)Y_i+\hat{Q}(0, W_i)(A_i-\hat{g}_i)}{1-\hat{\mathbb{E}}[A_i |W_i]}\Big)\Bigg] = \\
   \frac{1}{n}\sum_{i=1}^n \Big(\frac{A_i(y_i-\hat{Q}^1_i)}{\hat{g}_i} - \frac{(1-A_i)(y_i-\hat{Q}^0_i)}{1-\hat{g}_i}\Big) + 
    \frac{1}{n}\sum_{i=1}^n \big(\hat{Q}^1_i-\hat{Q}^0_i\big) = \\
   \frac{1}{n}\sum_{i=1}^n \Big(\frac{A_i(y_i-\hat{Q}^1_i)}{\hat{g}_i} - \frac{(1-A_i)(y_i-\hat{Q}^0_i)}{1-\hat{g}_i}\Big) + \hat{\beta}_{SR},
\end{multline}
where $\hat{Q}^k_i = \hat{Q}(k,W_i)=\hat{\mathbb{E}}[Y_i |A_i=k, W_i]$ and $\hat{g}_i = \hat{\mathbb{E}}[A_i |W_i]$. 

Among all the doubly robust estimators of ATE, AIPW is the most efficient estimator if both of the propensity score or outcome models are correctly specified, but is not necessarily efficient under incorrect model specification. In fact, this nice feature of AIPW may be less relevant in real life problems as we might not have a priori knowledge about the predictors of the propensity score and outcome and we cannot correctly model them. Further, in practice, perfect or near perfect prediction of the treatment assignment can inflate the variance of the AIPW estimator \citep{van2011targeted}. As a remedy, similar to the normalization of IPW estimator, we can define a normalized version of the AIPW estimator which is less sensitive to extreme values of predicted propensity score, referred to as the normalized Augmented Inverse Probability Weighting (nAIPW) estimator:

\begin{equation}\label{nAIPW}
    \hat{\beta}_{nAIPW} = \sum_{i=1}^n \Big(\frac{A_i(y_i-\hat{Q}^1_i)w_i^{(1)}}{\sum_{j=1}^n A_jw_j^{(1)}} - \frac{(1-A_i)(y_i-\hat{Q}^0_i)w_i^{(0)}}{\sum_{j=1}^n (1-A_j)w_j^{(0)}}\Big) + \hat{\beta}_{SR},
\end{equation}
where $w^{(1)}_k = \frac{1}{\hat{g}_k}$ and $w^{(0)}_k = \frac{1}{1-\hat{g}_k}$. Both AIPW and nAIPW estimators add adjustment factors to the SR estimator which involve both models of the treatment and the outcome. 

Both AIPW and nAIPW are examples of a class of estimators where

\begin{equation}\label{nAIPW1}
    \hat{\beta}_{GDR} = \frac{1}{n} \sum_{i=1}^n \Big(\frac{A_i(y_i-\hat{Q}^1_i)}{\hat{h}^1_i} - \frac{(1-A_i)(y_i-\hat{Q}^0_i)}{\hat{h}^0_i}\Big) + \hat{\beta}_{SR},
\end{equation}
where we refer to this general class as the General Doubly Robust estimator (GDR). Letting $\hat{h}^1=\hat{g}$ and $\hat{h}^0=1-\hat{g}$ gives the AIPW estimators and letting $\hat{h}^1=\hat{g}\hat{\mathbb{E}}\frac{A}{\hat{g}}$ and $\hat{h}^0=(1-\hat{g})\hat{\mathbb{E}}\frac{1-A}{1-\hat{g}}$ gives the nAIPW estimator.

The GDR estimator can also be written as
\begin{equation}
\hat{\beta}_{GDR} = \hat{\mathbb{E}} \Big(\big[\frac{A}{\hat{h}^1} - \frac{1-A}{\hat{h}^0}\big]y -
    \big(A-\hat{h}^1\big)\hat{Q}^1 + 
    \big(1-A-\hat{h}^0\big)\hat{Q}^0
    \Big),
\end{equation}

If $h^1$ and $h^0$ are chosen so that 
\begin{equation}
\mathbb{E}\big[A-h^1\big] = 0, \ \mathbb{E}\big[1-A-h^0\big] =0,
\end{equation}
by the total law of expectation $\hat{\beta}_{GDR}$ is an unbiased estimator of $\beta$.

\section{Outcome and Treatment Predictions}\label{prediction-model}

The causal estimation and inference when utilizing the AIPW and nAIPW is carried out in two steps. In step 1, the treatment and outcome are predicted by a statistical or machine learning (ML) algorithm, and in the second step the predictions are inserted into the estimator. The ML algorithms in step 1 can capture the linear and nonlinear relationships between the confounders and the treatment and the outcome.

Neural Networks (NNs) are a class of non-linear and non parametric complex algorithms that can be employed to model the relationship between any set of inputs and some outcome. There has been a tendency to use NNs as they have achieved a great success in the Artificial Intelligence (AI) most complex tasks such as computer vision and natural language understanding \citep{goodfellow2016deep}. 

\citet{farrell2018deep} used two independent NNs for modeling for modeling the propensity score model and the outcome with the Rectified Linear Unit (RELU) activation function \citep{goodfellow2016deep}, here referred to as the double NN or dNN:

\begin{equation} \label{dNNloss}
\begin{split}
        \mathbb{E}[Y|A, W] &= \beta_0 + \beta A + \mathbf{W}\alpha + \mathbf{H}\mathbf{\Gamma}_Y\\
        \mathbb{E}[A|W] &= \beta'_0 + \mathbf{W'}\alpha' + \mathbf{H'}\mathbf{\Gamma}_A,
\end{split}
\end{equation}
where two separate neural nets model $y$ and $A$ (no parameter sharing). \cite{farrell2018deep} proved that dNN algorithms almost attain $n^{\frac{1}{4}}$-rates. By employing the cross-fitting method and theory developed by Chernozhukov et al. \citep{chernozhukov2018double}, an orthogonal causal estimator is asymptotically normal, under some regularity and smoothing conditions. 
if dNN is used in the first step (see Theorem 1 in Farrell et al. paper \citep{farrell2018deep}).

These results assume no regularizations imposed on the NNs' weights, and only the Stochastic Gradient Descent (SGD) is used. Farrell et al. claim that the fact that SGD controls the complexity of the NN algorithm up to some extent \citep{goodfellow2016deep, zhang2021understanding} is sufficient for the first step. Our initial simulations, however, contradict this claim and we hypothesise that for causal parameter estimation, dNN with no regularization leads to high variance for the causal estimator used in the second step. Our initial experiments indicate that $L_2$ regularization and Dropout do not perform well in terms of the Mean Square Error (MSE) of AIPW. The loss functions we use contains $L_1$ regularization (in addition to SGD during the optimization):

\begin{equation} \label{eq1}
\begin{split}
        L_y(\mathcal{P}_y, \beta, \alpha) = & \sum_{i=1}^n \Big[y_i - \alpha' - \beta A_i -\mathbf{W}_i\alpha - H_i^T\mathbf{\Gamma}_Y \Big]^2 + C_{L_1}\sum_{\omega \in \mathcal{P}} |\omega|,\\
        L_A(\mathcal{P}_A, \alpha') = & \sum_{i=1}^n \Big[A_i \log\Big(g\big(H_i^T\mathbf{\Gamma}_A\big)\Big) + (1-A_i) \log\Big(1-g\big(H_i^T\mathbf{\Gamma}_A\big)\Big)\Big] + C'_{L_1}\sum_{\omega \in \mathcal{P}} |\omega|,
\end{split}
\end{equation}
where $C_{L_1}, C'_{L_1}$ are hyperparameters, that can be set before training or be determined by Cross-Validation, that can convey the training to pay more attention to one part of the output layer. The dNN can have an arbitrary number of hidden layers, or the width of the network ($\mathcal{HL}$) is another hyperparameter. For a 3-layer network, $\mathcal{HL}=[l_1,l_2, ...,l_{h}]$, where $l_j$ is the number neurons in layer $j$, $j=1, 2, ..., h$. $\mathcal{P}_y, \mathcal{P}_A$, are the connection parameters in the nonlinear part of the networks, with $\Omega$'s being shared for the two outcome and propensity models. Noted that the Gradient descent-type optimizations in the deep learning platforms (such as pytorch in our case) do not cause the NN parameters shrink to zero.

\section{GDR Estimator Properties}\label{theory}

In this section we will see that nAIPW \eqref{nAIPW} is doubly robust, that is, if either of the outcome or propensity scores models are $\sqrt{n}$-consistent, nAIPW will be consistent. Further, nAIPW is orthogonal \citep{chernozhukov2018double} and is asymptotically linear under certain assumptions and we calculate its asymptotic variance.

\subsection{Consistency and Asymptotic Distribution of nAIPW}

In Causal Inference, estimating the causal parameter and drawing inference on the parameter are two major tasks. Employing machine learning algorithm to estimate $Q$ and $g$ in \eqref{nAIPW} is a means to estimate and draw inference on causal parameter; the ultimate goal is the relationship between the treatment and the outcome. This allows people to use blackbox ML models with no explanation how these models have learned from the explanatory features. The question is if the consistency and asymtotic normality of the second step causal estimator are preserved if complex ML algorithms are utilized twice for the treatment and outcome models each with convergence rate of smaller than $\sqrt{n}$, and entropy that grows with $n$. 

\citet{chernozhukov2016double} provide numerical experiments illustrating that some estimators are not consistent or asymptotically normal if complex ML models are used that do not belong to the Donsker class and their entropy grow with $n$. They further provide a solution by introducing "orthogonal" estimators that, under some regulatory conditions and cross-fitting, are asymptotically normal even if complex ML models can be used as long as their rates of convergence are even as small as $n^{\frac{1}{4}}$. 

The next two subsections provide an overview of the general theory and prove that nAIPW is asymptotically normal.

\subsection{The Efficient Influence Function}

\citet{hines2021demystifying} derives the Efficient Influence Function (EIF) of $\beta=\beta_1-\beta_0$ as
\begin{equation}
    \phi(O, P) = \Big(\frac{A}{g}(Y-Q^1)+Q^1-\beta_1\Big) - \Big(\frac{1-A}{1-g}(Y-Q^0)+Q^0-\beta_0\Big)
\end{equation}


To study the asymptotic behaviour of nAIPW, we write the scaled difference
\begin{equation}\label{taylor}
    \sqrt{n}(\hat{\beta}-\beta) = \frac{1}{\sqrt{n}}\sum_{i=1}^n{\phi(O_i, P)}-\frac{1}{\sqrt{n}}\sum_{i=1}^n{\phi(O_i, \hat{P}_n)}+\sqrt{n}(P_n-P)[\phi(O_i, \hat{P}_n)-\phi(O_i, P)]-\sqrt{n}R(P, \hat{P}_n),
\end{equation}
where the first term is a normal distribution by the Central Limit Theorem, and third and forth terms are controlled if the class of functions are Donsker, and standard smoothing conditions are satisfied (\cite{chernozhukov2018double, van2000asymptotic}, Theorem 19.26). If the nuisance parameters are not Donsker, data splitting and cross fitting guarantees plus the regulatory conditions are needed to control these two terms \cite{chernozhukov2018double, farrell2018deep}. It is unclear, however, that how the second term behaves, i.e., 
\begin{equation}
    -\frac{1}{\sqrt{n}}\phi(O, \hat{P}_n) = -\frac{1}{\sqrt{n}}\sum_{i=1}^n\Big[\frac{A_i}{g_i}(Y_i-\hat{Q}_i^1)-\frac{1-A_i}{1-g_i}(Y_i-\hat{Q}_i^0)+\hat{Q}_i^1-\hat{Q}_i^0\Big]-\hat{\beta},
\end{equation}
where $\hat{\beta}=\beta(\hat{P}_n)$, as it contains data-adaptive nuisance parameters estimations. There are different tricks how to get rid of this term. One method is the one-step method in which we move this term to the left to create a new estimator which is exactly as the AIPW estimator with known propensity scores:

\begin{equation}\label{onestep}
    \sqrt{n}(\hat{\beta}+\frac{1}{n}\phi(O, \hat{P}_n)-\beta) = \sqrt{n}\Big(\frac{1}{n}\sum_{i=1}^n\Big[\frac{A_i}{g_i}(Y_i-\hat{Q}_i^1)-\frac{1-A_i}{1-g_i}(Y_i-\hat{Q}_i^0)+\hat{Q}_i^1-\hat{Q}_i^0\Big]-\beta\Big).
\end{equation}

Another trick is to let this term vanish which result in the estimating equations whose solution is exactly the same as the one-step estimator. The targetted learning strategy is to manipulate the data generating process which results in a different estimator \cite{hines2021demystifying, van2011targeted} (which we do not study here). 

The requirement in the above estimator is that the propensity score is know which is unrealistic. In reality, this quantity should be estimated using the data. However, replacing $g$ with a data-adaptive estimator changes the remainder term in \eqref{taylor} that needs certain assumptions to achieve asymptotic properties such as consistency. We replace $g$ and $1-g$ in \eqref{onestep} by $\hat{h}^1$ and $\hat{h}^0$, respectively, which provides a more general view of the above one-step estimator.

\subsection{Doubly Robustness and  Rate Doubly Robustness Properties of GDR}\label{doublerobustnAIPW}

One of the appealing properties of AIPW is its doubly robust which partially relaxes the restrictions of IPW and SR which require the consistency of the treatment and outcome models, respectively. This property is helpful when the first step algorithms are $\sqrt{n}$-consistent. The following theorem states that the nAIPW estimator \eqref{nAIPW} actually possesses the doubly robustness property.

\begin{theorem}[nAIPW Double Robustness]\label{nAIPWrobustness}
The $DR$ estimator \eqref{nAIPW} is consistent if either $\hat{Q}^k\xrightarrow{p}Q^k$, $k=0, 1$, or $\hat{g}\xrightarrow{p}g$.
\end{theorem}

The proof is left to the appendix. 
Theorem \ref{nAIPWrobustness} is useful when we \emph{a peiori} knowledge about the propensity scores (such as in the experimental studies) or we estimate the propensity scores with $\sqrt{n}$-rate converging algorithms. In practice, however, the correct specification is infeasible in the observational data, but $\sqrt{n}$-rate algorithms such as parametric models, Generalized Additive Models (GAM) or the models that assume sparsity might be used \citep{farrell2015robust}. This is restrictive and these model assumptions might not hold in practice which is why non-parametric ML algorithms such as NNs are motivated to be used. As mentioned before, NN we utilize here does not offer a $\sqrt{n}$-consistent prediction model in the first step of the estimation \citep{farrell2018deep}. This reduces the usefulness of the double robustness property of GDR estimator when using complex ML algorithms. A more useful property when using complex ML algorithms is the \emph{rate double robustness (RDR)} property \cite{smucler2019unifying}. RDR does not require either of the prediction models to be $\sqrt{n}$-consistent; it suffices that they are consistent at any rate but together become $\sqrt{n}$-consistent; that is, if the propensity score and outcome model are consistent at $n^{r_A}$ and $n^{r_Y}$, respectively ($r_Y, r_A<0$), we must have $r_A+r_Y=\frac{1}{2}$. To see that the DR has this property (as does DR \citep{farrell2015robust}), note that the remainder \eqref{taylor} can be written as

\begin{equation}\label{remainder}
    -\sqrt{n}R(P, \hat{P}_n) = \sqrt{n}\mathbb{E}\Big[\big(\frac{g}{\hat{h}^1}-1\big)\big(Q^1-\hat{Q}^1\big)\Big]+\sqrt{n}\mathbb{E}\Big[\big(\frac{1-g}{\hat{h}^0}-1\big)\big(Q^0-\hat{Q}^0\big)\Big],
\end{equation}
which, by the H\"older inequality, is upper bounded:
\begin{equation}\label{bounded}
    -\sqrt{n}R(P, \hat{P}_n)\leq \Bigg[\mathbb{E}\Big[\frac{g}{\hat{h}^1}-1\Big]^2\Bigg]^{\frac{1}{2}}\Bigg[\mathbb{E}\Big[Q^1-\hat{Q}^1\Big]^2\Bigg]^{\frac{1}{2}}+\Bigg[\mathbb{E}\Big[\frac{1-g}{\hat{h}^0}-1\Big]^2\Bigg]^{\frac{1}{2}}\Bigg[\mathbb{E}\Big[Q^0-\hat{Q}^0\Big]^2\Bigg]^{\frac{1}{2}}
\end{equation}

Making the standard assumptions that
\begin{equation}\label{assumptions}
\begin{split}
    &\Bigg[\mathbb{E}\Big[g-\hat{h}^k\Big]^2\Bigg]^{\frac{1}{2}}\Bigg[\mathbb{E}\Big[Q^k-\hat{Q}^k\Big]^2\Bigg]^{\frac{1}{2}} = o(n^{-\frac{1}{2}}),\quad k=0,1,\\
    &\mathbb{E}\Big[g-\hat{h}^k\Big]^2 = \ o(1), \quad
    \mathbb{E}\Big[Q^k-\hat{Q}^k\Big]^2 = \ o(1), \quad k=0, 1,\\
    &\text{Empirical Positivity} \quad  c_1< \hat{h}^k < 1- c_2, \  \text{for some}\ c_1, c_2>0,
\end{split}
\end{equation}
implies
\begin{equation}
    -\sqrt{n}R(P, \hat{P}_n) = o(n^{-\frac{1}{2}}),
\end{equation}
that is the GDR has the rate double robustness property.

The assumptions in \eqref{assumptions} are less restrictive than needing at least one of the predictions models to be $\sqrt{n}$-consistent for the double robust property \citep{hines2021demystifying, farrell2015robust}. This means that the outcome and propensity score models can be at least as fast as $o(n^{-\frac{1}{4}})$ (which is an attainable generalization bound for many complex machine learning algorithms \citep{chernozhukov2018double}), and still the GDR estimator is consistent. \citet{farrell2018deep} proves that two neural networks without regularization (except the one imposed by the stochastic gradient descent optimization) satisfy such bounds and can provide a convenient first-step prediction algorithms (while they utilize AIPW estimator and the cross-fitting strategy proposed by \citet{chernozhukov2018double}).

In order for a special case of GDR estimator to outperform the AIPW estimator, we must have $Ah^1 \geq Ag$ and $(1-A)h^0 \geq (1-A)(1-g)$, in addition to conditions in \eqref{assumptions}. Noted that these two conditions are satisfied for nAIPW; replacing $h^1$ and $h^0$ with $\hat{g}\hat{\mathbb{E}}\frac{A}{\hat{g}}$ and $(1-\hat{g})\hat{\mathbb{E}}\frac{1-A}{1-\hat{g}}$ can help stabilize the bias and variance magnitude and help shrink the remainder \eqref{remainder} to zero. The scenario analysis performed in Subsection \ref{scenarioanalysis} provides an insight about the reduction to the sensitivity to the violation of empirical positivity assumption.

\subsection{Robustness of nAIPW against Extreme Propensity scores}\label{scenarioanalysis}

There are two scenarios that the positivity is violated, where the probability of receiving the treatment for the people who are treated is 1, that is, $A_k=1$ and $P(A_k=1|W)=1$ (or vise versa for the untreated group $A_k=0$ and $P(A_k=0|W)=0$), and where there are a handful of treated subjects whose probability of receiving the treatment is 0, that is, $A_k=1$ and $P(A_k=1|W)=0$ (and vise versa for the untreated group, that is, $A_k=0$ and $P(A_k=0|W)=1$). Although the identifiability assumptions guarantee that such scenarios do not occur, in practice extremely small or large probabilities similar to the second scenario above, that is where there exists a treated individual who has a near zero probability of receiving the treatment, can impact the performance of the estimators that involve propensity score weighting. For example, replacing $h^1$ with $\hat{g}$ and $h^0$ with $1-\hat{g}$ in practice can increase both bias and variance of AIPW \citep{van2011targeted}. This can be seen by veiwing the bias and variance of these weighting terms. As noted before, the AIPW and nAIPW add adjustments to the Single robust estimator $\mathbb{E}Q^1-Q^0$. The adjustments involve weightings $\frac{A}{g}$ or $\frac{A}{g\mathbb{E}\frac{A}{g}}$ to the residuals of $Y$ and $Q^k$, $k=0, 1$. Under a correct specification of the propensity score $g$, these weights have the same expectations. The difference is in their variances:
\begin{equation}
    \begin{split}
        &Var(\frac{A}{g})=\frac{1}{g}-1,\\
        &Var \Big(\frac{A}{g\mathbb{E}\frac{A}{g}}\Big)=\frac{1}{\mathbb{E}^2\frac{A}{g}}(\frac{1}{g}-1),
    \end{split}
\end{equation}
under the correct specification of the propensity score $g$. By letting $g$ tend to zero in near violation of positivity assumption, it can be seen that the nAIPW is less volatile than AIPW estimator. That is, the weights in AIPW might have large variance that those in nAIPW.

A scenario analysis is performed to see how nAIPW stabilizes the estimator: Assume that the empirical positivity is near violated, that is there is at least an observation $k$ where $A_k=1$ where $\hat{g}_k$ is extremely close to zero, say $\hat{g}_k=10^{-s}$ for $s\gg 0$. AIPW will blow up in this case:
\begin{equation}\label{scen11}
\begin{split}
    \beta_{1, AIPW} &= \frac{1}{n}\Big(10^s(Y^1_k-Q^1_k)+\sum_{i\in I^1_{-k} }\frac{Y^1_i-Q^1_i}{g_i}\Big)+\frac{1}{n}\sum_{i=1}^nQ^1_i,\\
    \beta_{0, AIPW} &= \frac{1}{n}\Big(\sum_{i\in I^0 }\frac{Y^0_i-Q^0_i}{1-g_i}\Big)+\frac{1}{n}\sum_{i=1}^nQ^0_i,\\
\end{split}
\end{equation}
where $I^a=\{j:\ A_j=a\}$, $I^a_{-k} = \{j:\ A_j=a\}$, and subscripts $a=1$ and $a=0$ refer to the  estimators of the first and the second components in ATE \eqref{causalriskdiff0}. However, nAIPW is robust against this positivity violation: 
\begin{equation}\label{scen1-1}
    \beta_{1, nAIPW} = \Big(\frac{Y^1_k-Q^1_k}{10^{-s}(10^{s}+\sum_{j\neq k}{\frac{A_j}{g_j}})}+\sum_{i\in I^1_{-k} }\frac{Y^1_i-Q^1_i}{g_i(10^{s}+\sum_{j\neq k}{\frac{A_j}{g_j}})}\Big)+\frac{1}{n}\sum_{i=1}^nQ^1_i,
\end{equation}
and
\begin{equation}
    \beta_{0, nAIPW} = \Big(\frac{0\times(Y^1_k-Q^0_k)}{\star}+\sum_{i\in I^0_{-k} }\frac{Y^0_i-Q^0_i}{(1-g_i)(\sum_{j=1}^n{\frac{1-A_j}{1-g_j}})}\Big)+\frac{1}{n}\sum_{i=1}^nQ^0_i.
\end{equation}

Thus
\begin{equation}\label{scen11-}
    \beta_{1, nAIPW} \approx \Big(\frac{Y^1_k-Q^1_k}{1+10^{-s}(n-1)}+\sum_{i\in I^1_{-k} }\frac{Y^1_i-Q^1_i}{g_i10^{s}+g_i(n-1)}\Big)+\frac{1}{n}\sum_{i=1}^nQ^1_i ,
\end{equation}

The factor $10^s$ in \eqref{scen11} can blow up the AIPW if $10^s \gg n$ (and the outcome estimation is not enough close to the observer outcome), but this factor does not appear in the numerator of nAIPW estimator. For such large factors, \eqref{scen11-} can be simplified to
\begin{equation}\label{scen11---}
    \beta_{1, nAIPW} \approx Y^1_k-Q^1_k+\frac{1}{n}\sum_{i=1}^nQ^1_i.
\end{equation}
Thus the extreme probability does not make $\beta_{1, nAIPW}$ blow up, but the adjustment to the $\beta_{1, SR}$ that accounts for confounding effects. The second factor $\beta_{0, nAIPW}$ is not impacted in this scenario.

Considering a scenario that there is another treated individual with extremely small probability, say $\hat{g}_l=10^{-t}$ such that, without loss of generality, $t>s\gg0$, we will have:
\begin{equation}\label{scen12}
    \beta_{1, nAIPW} \approx \frac{Y^1_k-Q^1_k}{1+10^{t-s}+10^{-s}(n-2)}+\frac{Y^1_l-Q^1_l}{1+10^{s-t}+10^{-t}(n-2)}+\frac{1}{n}\sum_{i=1}^nQ^1_i.
\end{equation}
Depending on the values $s$ and $t$ one of the first two terms in \eqref{scen12} might vanish, but the estimator does not blow up. There is at most only a handful of treated individuals with extremely small probabilities and based on the above observation, the nAIPW estimator does not blow up. That said, nAIPW might not sufficiently correct the $\beta_{SR}$ for the confounding effects, although, confounders have been taken into account in the calculation of $\beta_{SR}$ up to some extent. 


The same observation can be made in the asymptotic variance of these estimators. This shows how extremely small probabilities for treated individuals (or extremely large probabilities for un-treated individuals) can result in a biased and unstable estimator, while neither of the bias or variance of nAIPW suffer as much. Although not performed, the same observation can be made for the untreated individuals with extremely large probabilities. 

As the above scenario analysis indicates the bias and variance of nAIPW might go up in cases of near violation of positivity, but it still is less biased and more stable than AIPW. The remainder term \eqref{remainder} is also more likely to be $o(n^{-\frac{1}{2}})$ in nAIPW versus AIPW because for $k$'s that $A_k=1$, $g_k\mathbb{E}_n\frac{A_k}{g_k}\geq g_k$.

\subsection{Asymptotic Sampling Distribution of nAIPW}

Replacing $g$ in the denominator of the von Mises expansion \eqref{taylor} with the normalizing terms is enough to achieve the asymptotic distribution of the nAIPW and its asymptotic standard error. However, we can see that nAIPW is also the solution to an (extended) estimating equations. Solution to the estimating equations is important as Van der Vaart (Chapters 19 and 25) proves that under certain regulatory conditions, if the prediction models belong to the Donsker class, the solution to Z-estimators are consistent and asymptotically normal (\cite{van2000asymptotic}, Theorem 19.26). Thus, nAIPW that is the solution to a Z-estimator (also referred to an M-estimator) will inherit the consistency and asymptotically normality, assuming certain regulatory conditions and that the first-step prediction models belong to the Donsker class:
\begin{equation} \label{est-eqs}
\begin{split}
        \mathbb{E} \Bigg[\frac{A(Y^1 - Q^1)}{\gamma g} - \frac{(1-A)(Y^0 - Q^0)}{\lambda(1-g)} + (Q^1 - Q^0 - \beta)\Bigg] &= 0,\\
        \mathbb{E} \Big[\frac{A}{g} - \gamma\Big]&= 0,\\
        \mathbb{E} \Big[\frac{1-A}{1-g} - \lambda\Big] &= 0.
\end{split}
\end{equation}

The Donsker class assumption prevents too complex algorithms in the first step, algorithms such as Tree-based models, NNs, Cross-hybrid algorithms or their aggregations \citep{hines2021demystifying, friedman2001elements}. The Donsker class assumption can be relaxed if sample splitting (or cross-fitting) is utilized and the target parameter is orthogonal \citep{chernozhukov2018double}. In the next section we see that nAIPW is orthogonal, and thus, theoretically, we can relax the Donsker class assumption under certain smoothing regulatory conditions. Before seeing the orthogonality property of nAIPW, let us review the smoothing regularity conditions necessary for asymptotic normality. Let $\beta$ be the causal parameter, $\eta \in T$ be the infinite dimensional nuisance parameters where $T$ is a convex set with a norm. Also let the score function $\phi: \mathbb{O} \times \mathcal{B} \times T \rightarrow \mathbb{R}$ be a measurable function, and $\mathbb{O}$ be the measurable space of all random variables $O$ with probability distribution $P \in \mathcal{P}_n$ and $\mathcal{B}$ is an open subset of $\mathbb{R}$ containing the true causal parameter. Let the sample $O=(O_1, O_2, ..., O_n)$ be observed and the set of probability measures $\mathcal{P}_n$ can expand with sample size $n$. Also let $\beta \in \mathcal{B}$ is the solution to the estimating equation $\mathbb{E}\phi(\mathbb{O}, \beta, \eta)=0$. The assumptions that guarantee the second-step orthogonal estimator $\hat{\beta}$ be asymptotically normal are \cite{chernozhukov2018double}: 1) $\beta$ does not fall on the boundary of $\mathcal{B}$; 2) The map $(\beta, \eta) \rightarrow \mathbb{E}_P \phi\big(\mathbf{O}, \beta, \eta\big)$ is twice Gateauax differentiable (this holds bythe positivity assumption).
$\beta$ is identifiable; 3) $\mathbb{E}_P \phi\big(\mathbf{O}, \beta, \eta\big)$ is smooth enough; 4) $\hat{\eta}\in \mathcal{T}$ with high probability and $\eta\in \mathcal{T}$.
$\hat{\eta}$ converges to $\eta_0$ at least as fast as $n^{-\frac{1}{4}}$ (similar but slightly stronger than first two assumptions in \eqref{assumptions}); 5) Score function(s) $\phi(., \beta, \eta)$ has finite second moment for all $\beta \in \mathcal{B}$ and all nuisance parameters $\eta \in \mathcal{T}$; 6) The score function(s) $\phi(., \beta, \eta)$ is measurable; 7) The number of folds increases by sample size.

\subsection{Orthogonality and the Regulatory Conditions}\label{orthogonality}

The Orthogonality condition \citep{chernozhukov2018double} is a property related to the estimating equations
\begin{equation}\label{esteq}
   \mathbb{E} \phi(\mathbf{O}, \beta, \eta) = 0.
\end{equation}
We refer to an estimator drawn from the estimating equations \eqref{esteq} as an orthogonal estimator.

Let $\eta \in T$, where $T$ be a convex set with a norm. Also let the score functions $\phi: \mathbb{O} \times \mathcal{B} \times T \rightarrow \mathbb{R}$ be a measurable function\footnote{For a higher but finite dimensional causal parameter, the score function is a vector of measurable functions.}, and $\mathbb{O}$ is measurable space of all random variables $O$ with probability distribution $P \in \mathcal{P}_n$ and $\mathcal{B}$ is an open subset of $\mathbb{R}$ containing the true causal parameter. Let the sample $O=(O_1, O_2, ..., O_n)$ be observed and the set of probability measures $\mathcal{P}_n$ can expand with sample size $n$. 
The score function $\phi$ follows the Neyman Orthogonality condition with respect to $\mathcal{T} \subseteq T$, if the Gateauax derivative operator exists for all $\epsilon \in [0, 1)$:

\begin{equation}\label{orthcond}
   \partial_{\tilde{\eta}} \mathbb{E}_P \phi\big(\mathbf{O}, \beta_0, \tilde{\eta}\big)\Big|_{\tilde{\eta}=\eta} [\tilde{\eta}-\eta] := \partial_{\epsilon} \mathbb{E}_P \phi\big(\mathbf{O}, \beta_0, \eta+\epsilon(\tilde{\eta}-\eta)\big)\Big|_{\epsilon=0} = 0.
\end{equation}

\cite{chernozhukov2016double} present a few examples of orthogonal estimating equations including AIPW estimator \eqref{AIPW}. 

Utilizing cross-fitting, under standard regulatory conditions, the asymptotic normality of estimators with orthogonal estimating equations is guaranteed even if the nuisance parameters are estimated by ML algorithms not belonging to the Donsker class and without finite entropy conditions \citep{chernozhukov2016double}. The regulatory conditions to be satisfied are 1) $\beta$ does not fall on the boundary of $\mathcal{B}$; 2) The map $(\beta, \eta) \rightarrow \mathbb{E}_P \phi\big(\mathbf{O}, \beta, \eta\big)$ is twice Gateauax differentiable.
$\beta$ is identifiable; 3) $\mathbb{E}_P \phi\big(\mathbf{O}, \beta, \eta\big)$ is smooth enough; 4) $\hat{\eta}\in \mathcal{T}$ with high probability and $\eta\in \mathcal{T}$.
$\hat{\eta}$ converges to $\eta_0$ at least as fast as $n^{-\frac{1}{4}}$; 5) Score function(s) $\phi(., \beta, \eta)$ has finite second moment for all $\beta \in \mathcal{B}$ and all nuisance parameters $\eta \in \mathcal{T}$; 6) The score function(s) $\phi(., \beta, \eta)$ is measurable; 7) The number of folds increases by sample size.

Replacing $\lambda$ and $\gamma$ in the first line of \eqref{est-eqs} with their solutions in the second and third equations:

\begin{equation} \label{est-eq}
        \mathbb{E}_P \phi\big(\mathbf{O}, \beta, Q^1, Q^0, g\big) = \mathbb{E} \Bigg[\frac{A(Y^1 - Q^1)}{g\mathbb{E}\frac{A}{g}} - \frac{(1-A)(Y^0 - Q^0)}{(1-g)\mathbb{E}\frac{1-A}{1-g}} + (Q^1 - Q^0 - \beta)\Bigg] = 0,
\end{equation}

Implementing the orthogonality condition \eqref{orthcond}, it can be verified that nAIPW \eqref{nAIPW} is also an example of orthogonal estimators. To see this, we apply the definition of orthogonality \citep{chernozhukov2018double}:
\begin{multline}\label{nAIPWorthog}
   \partial_{\eta} \mathbb{E}_P \phi\big(\mathbf{O}, \beta, \eta\big)\Big|_{\eta=\eta_0} [\eta-\eta_0] = \partial_{\eta} \mathbb{E}_P\Big(Q^1 + \frac{A(Y^1-Q^1)}{g\mathbb{E}\frac{A}{g}} - Q^0 - 
   \frac{(1-A)(Y^0-Q^0)}{(1-g)\mathbb{E}\frac{1-A}{1-g}}-\beta\Big)|_{\eta=\eta_0} [\eta-\eta_0]\propto\\ \partial_{\epsilon} \mathbb{E}_P\Big(Q^1_{\epsilon} + \frac{A(Y^1-Q_{\epsilon}^1)}{g_{\epsilon}\mathbb{E}\frac{A}{g_{\epsilon}}} - Q^0_{\epsilon} - 
   \frac{(1-A)(Y^0-Q^0_{\epsilon})}{(1-g_{\epsilon})\mathbb{E}\frac{1-A}{1-g_{\epsilon}}}-\beta\Big)|_{\epsilon=0} =\\ \mathbb{E}\Big((\tilde{Q}^1-Q^1)+\frac{A}{g\mathbb{E}\frac{A}{g}}\big(-(\tilde{Q}^1-Q^1)\big)+A(Y-Q^1)a(g,\tilde{g}-g)\Big)-\\
   \mathbb{E}\Big((\tilde{Q}^0-Q^0)+\frac{1-A}{(1-g)\mathbb{E}\frac{1-A}{1-g}}\big(-(\tilde{Q}^0-Q^0)\big)+(1-A)(Y-Q^0)b(g,\tilde{g}-g)\Big)=0, 
\end{multline}
where $Q^k_{\epsilon}=\epsilon \tilde{Q}^k+ (1-\epsilon)Q^k$, $k=0,1$, and $g_{\epsilon}=\epsilon \tilde{g} + (1-\epsilon)g$, and for some functions $a$, and $b$. The last equality is because $\mathbb{E}A(Y-Q^1)=0$, $\mathbb{E}(1-A)(Y-Q^0)=0$, $\mathbb{E}\frac{A}{g\mathbb{E}\frac{A}{g}}=1$ and $\mathbb{E}\frac{1-A}{(1-g)\mathbb{E}\frac{1-A}{1-g}}=1$, under correct specification of the propensity score $g$.

Thus, nAIPW is orthogonal, and by utilizing cross-fitting for the estimation nAIPW is consistent and asymptotically normal, under certain regulatory conditions.

\subsection{Asymptotic Variance of nAIPW}\label{asymptnAIPW}

To evaluate the asymptotic variance of nAIPW, we employ the M-estimation theory \citep{van2000asymptotic, stefanski2002calculus}. For Causal Inference for M-estimators, the bootstrap for the estimation of causal estimator variance is not generally valid even if the nuisance parameter estimators are $\sqrt{n}$-convergent. However, subsampling $m$ out of $n$ observations \citep{politis1994stationary} can be shown to be universally valid, provided $m \rightarrow \infty$ and $\frac{m}{n} \rightarrow 0$. In practice, however, we can face computational issues since nuisance parameters must be separately estimated (possibly with ML models) for each subsample/bootstrap sample.

The variance estimator of AIPW \eqref{AIPW} is \citep{lunceford2004stratification}
\begin{multline}\label{varAIPW}
    \hat{\sigma}^2_{AIPW} = \frac{1}{n^2}\sum_{i=1}^n \Big(\frac{A_iY_i-\hat{Q}^1_i(A_i-\hat{g}_i)}{\hat{g}_i} -    \frac{(1-A_i)Y_i+\hat{Q}^0_i(A_i-\hat{g}_i)}{1-\hat{g}_i} - \hat{\beta}_{AIPW}\Big)^2=\\
    \frac{1}{n^2}\sum_{i=1}^n \Big(\frac{A_i(y_i-\hat{Q}^1_i)}{\hat{g}_i}-\frac{(1-A_i)(y_i-\hat{Q}^0_i)}{1-\hat{g}_i} + \hat{\beta}_{SR} - \hat{\beta}_{AIPW} \Big)^2.
\end{multline}

The below theorem states that the variance estimator of AIPW \eqref{varAIPW} can intuitively extend to calculate the variance estimator of nAIPW \eqref{nAIPW} by moving the denominator $n^2$ to the square term in the summation and replace it with $\hat{g}\hat{\mathbb{E}}\big(\frac{A}{\hat{g}}\big)$ or $(1-\hat{g})\hat{\mathbb{E}}\big(\frac{1-A}{1-\hat{g}}\big)$ in the terms containing $g$ and $1-g$ in the denominator, respectively.

\begin{theorem}\label{nAIPWvartheorem}

The asymptotic variance of the nAIPW \eqref{nAIPW} is

\begin{equation}\label{nAIPWvar}
    \hat{\sigma}^2_{nAIPW} = \sum_{i=1}^n \Big(
    \frac{A_i(y_i-\hat{Q}^1_i)w_i^{(1)}}{\sum_{j=1}^n A_jw_j^{(1)}} -    \frac{(1-A_i)(y_i-\hat{Q}^0_i)w_i^{(0)}}{\sum_{j=1}^n (1-A_j)w_j^{(0)}} + \frac{1}{n}\big(\hat{\beta}_{SR} - \hat{\beta}_{nAIPW}\big) \Big)^2,
\end{equation}
where $\hat{Q}^k_i = \hat{Q}(k,W_i)$ and $\hat{g}_i = \hat{\mathbb{E}}[A_i |W_i]$. 
\end{theorem}
The proof utilizing the estimating equation techbique is straightforward and is left to the Appendix \ref{appendix}. The same result can be seen when deriving the estimator in the one-step method (see \eqref{taylor} and \eqref{onestep}). The above theorem states that the variance estimator of AIPW \eqref{varAIPW} can intuitively extend to calculate the variance estimator of nAIPW \eqref{nAIPW} by moving the denominator $n^2$ to the square term in the summation and replace it with $\hat{g}\hat{\mathbb{E}}\big(\frac{A}{\hat{g}}\big)$ or $(1-\hat{g})\hat{\mathbb{E}}\big(\frac{1-A}{1-\hat{g}}\big)$ in the terms containing $g$ and $1-g$ in the denominator, respectively. This is intuitive because, by the Law of Total Probability, $\mathbb{E}$ the first two terms is $n$. 

\section{Monte Carlo Experiments}

A Monte Carlo simulation study (with 100 iterations) was performed to compare AIPW and nAIPW estimators, where the dNN is used for the first step prediction. There are a total of 2 case scenarios according to the size of the data. We fixed the sample sizes to be  $n = 750$ and $n = 7500$ , with the number of covariates $p = 32$ and $p = 300$, respectively. The predictors include four types of covariates: the confounders, $X_c $, instrumental variables, $X_{iv}$, the outcome predictors, $X_y$ and the noise or irrelevant covariates $X_{irr}$. Their sizes for the scenarios are $\#X_c  = \#X_{iv} = \#X_y = \#X_{irr} = 8, 75$ and independent from each other were drawn from the Multivariate Normal (MVN) Distribution as $X \sim \mathcal{N}(\mathbf{0}, \Sigma)$, with $\Sigma_{kj} = \rho^{j-k}$ and $\rho=0.5$. The models to generate the treatment assignment and outcome were specified as
\begin{equation} \label{generatedata}
\begin{split}
    A &\sim Ber(\frac{1}{1+e^{-\eta}}), \text{with} \ \eta = f_a(X_c)\gamma_c + g_a(X_{iv})\gamma_{iv},\\
    y &= 3 + A + f_y(X_c)\gamma'_c + g_y(X_y)\gamma_y + \epsilon,
\end{split}
\end{equation}

and $\beta = 1$. The functions $f_a, g_a, f_y, g_y$ select 20\% of the columns and apply interactions and non-linear functions listed below \eqref{nonlinearfs}. The strength of instrumental variable and confounding effects were chosen as $\gamma_c, \gamma_c', \gamma_y \sim Unif(r_1, r_2)$ where  $(r_1=r_2=0.25)$, and $\gamma_{iv} \sim Unif(r_3, r_4)$ where $(r_3=r_4=0.25)$.

The non-linearities are randomly selected among the following functions:
\begin{equation} \label{nonlinearfs1}
\begin{split}
    & l(x_1, x_2) = e^{\frac{x_1x_2}{2}}\\
    & l(x_1, x_2) = \frac{x_1}{1+e^{x_2}}\\
    & l(x_1, x_2) = \big(\frac{x_1x_2}{10}+2\big)^3\\
    & l(x_1, x_2) = \big(x_1+x_2+3\big)^2\\
    & l(x_1, x_2) = g(x_1) \times h(x_2)
\end{split}
\end{equation}
where $g(x) = -2 I(x \leq -1) - I(-1 \leq x \leq 0) + I(0 \leq x \leq 2)+ 3 I(x \geq 2)$, and $h(x) = -5 I(x \leq 0) - 2 I(0 \leq x \leq 1) + 3 I(x \geq 1)$, or $g(x) = I(x \geq 0)$, and $h(x) = I(x \geq 1)$.

The networks' activation function is Rectified Linear Unit (ReLU), with 3 hidden layers as large as the input size (p), with $L_1$ regularization and batch size equal to $3*p$ and 200 epochs. The Adaptive Moment Estimation (Adam) optimizer \citep{kingma2014adam} with learning rate 0.01 and momentum 0.95 were used to estimate the network's parameters, including the causal parameter (ATE).

\section{Simulation Results}\label{simulations}

The oracle estimations are plotted in all the graphs to compare the real-life situations with the truth. In almost all the scenarios we cannot obtain perfect causal effect estimation and inference.

Figure \ref{boxplot1} shows the distribution of AIPW and nAIPW for different hyperparameter settings of NNs. The nAIPW estimator outperforms AIPW in almost all the scenarios. As the AIPW give huge values in some simulation iterations the log of the estimation is taken in Figure \ref{boxplot1}.


\begin{figure}[ht!]
\caption{The distribution of log of the estimated AIPW and nAIPW in the 100 simulated iterations. The performance of nAIPW is clearly superior to the performance of AIPW as it is less dispersed and has more stable in terms of different hyperparameter settings. $p$ is either 32 or 300 for the small or large datasets and $q\approx\frac{p}{10}$, that is 3 or 30.}
\includegraphics[scale=0.8]{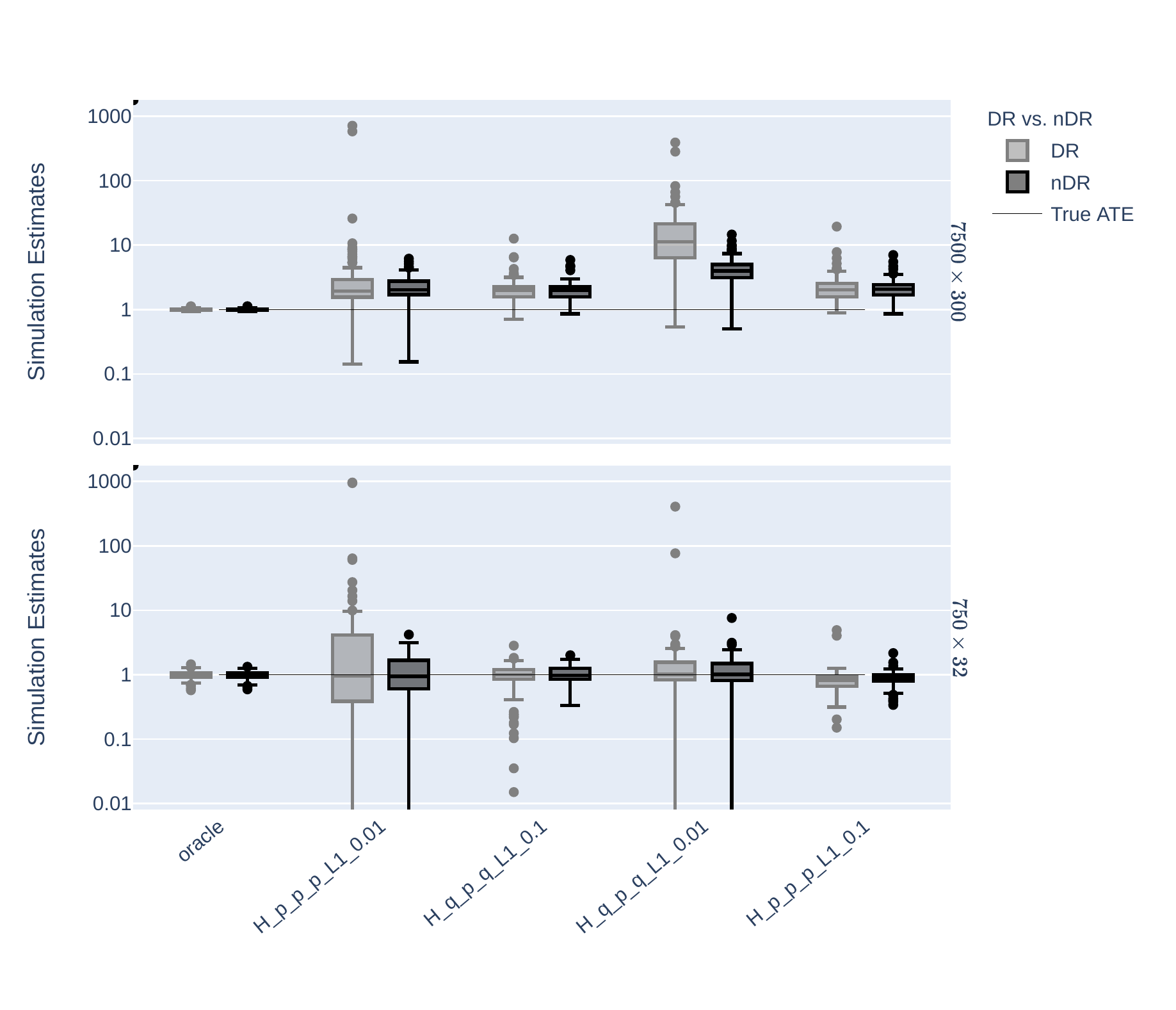}
\label{boxplot1}
\end{figure}

We also compare the estimators in different scenarios with bias, variance and their trafeoff measures:

\begin{equation} \label{nonlinearfs}
\begin{split}
    &\text{Bias} \quad \quad \hat{\delta} = \beta - \frac{1}{m}\sum_{j=1}^{m} \hat{\beta}_j\\
    &\text{MC std} \quad \quad \hat{\sigma}_{MC} = \sqrt{\frac{1}{m}\sum_{j=1}^{m} (\hat{\beta}_j-\hat{\mu})^2}\\
    &\text{MC RMSE} \quad \quad RMSE = \sqrt{\hat{\sigma}_{MC}^2 + \hat{\delta}^2}\\
    &\text{Asymptotic StdErr} \quad \quad \hat{\sigma}_{SE} = \frac{1}{m}\sum_{j=1}^{m}\hat{\sigma}_j,
\end{split}
\end{equation}
where $\beta=1$, with $\hat{\beta}_j$'s are the AIPW or nAIPW estimations in the $j^{\textit{th}}$ simulation round, $\hat{\mu} = \frac{1}{m}\sum_{j=1}^{m}\hat{\beta}_j$ and  $m=100$ is the number of simulation rounds, and $\hat{\sigma}$ is square root of \eqref{varAIPW} or \eqref{nAIPWvar}.

Figure \ref{nAIPWvsAIPW} demonstrates the bias, MC Standard Deviation (MC std) and the Root Mean Square Error (RMSE) of AIPW and nAIPW estimators for the scenarios where $n=750$ and $n=7500$, and for 4 hyperparameter sets ($L_1$ regularization and width of the dNN). In general, in each figure of the panel, the hyperparameter scenarios in the left imply more complex model (with less regularization or narrower network). In these graphs, the lower the values, the better the estimator. For the smaller data size $n=750$ in the left 3 panels, the worst results are attributed to AIPW when there is least regularization and the hidden layers are as wide as number of inputs. To have more clear plots for comparison, we have skipped plotting the upper bounds as they were large numbers; the lower bounds are enough to show the significance of the results. In the scenarios where there are smaller number of hidden neurons with 0.01 $L-1$ regularization, the bias, variance and their trade-off, RMSE, are more stable. By increasing the $L_1$ regularization, these measure go down which indicates the usefulness of regularization and AIPW normalization for causal estimation and inference. Almost the same pattern is seen for larger size ($n=7500$) scenario, except the bump in all the three measures in the hyperparameter scenario where regularization remains the same ($L_1=0.01$) and the number of neurons in the first and last hidden layers are small too. In all 3 measures of bias, standard deviation and RMSE, nAIPW is superior to AIPW, or at least there is no statistically significant difference between AIPW and nAIPW. 




\begin{figure}[ht!]
\centering
\caption{The bias, MC standard error and the root mean square error of the AIPW and nAIPW estimators for different data sizes and NN hyperparameters ($L_1$ regularization and width of the network.) $p$ is either 32 or 300 for the small or large datasets and $q\approx\frac{p}{10}$, that is 3 or 30. The estimates are capped at -10 and 10.}
\includegraphics[scale=0.7]{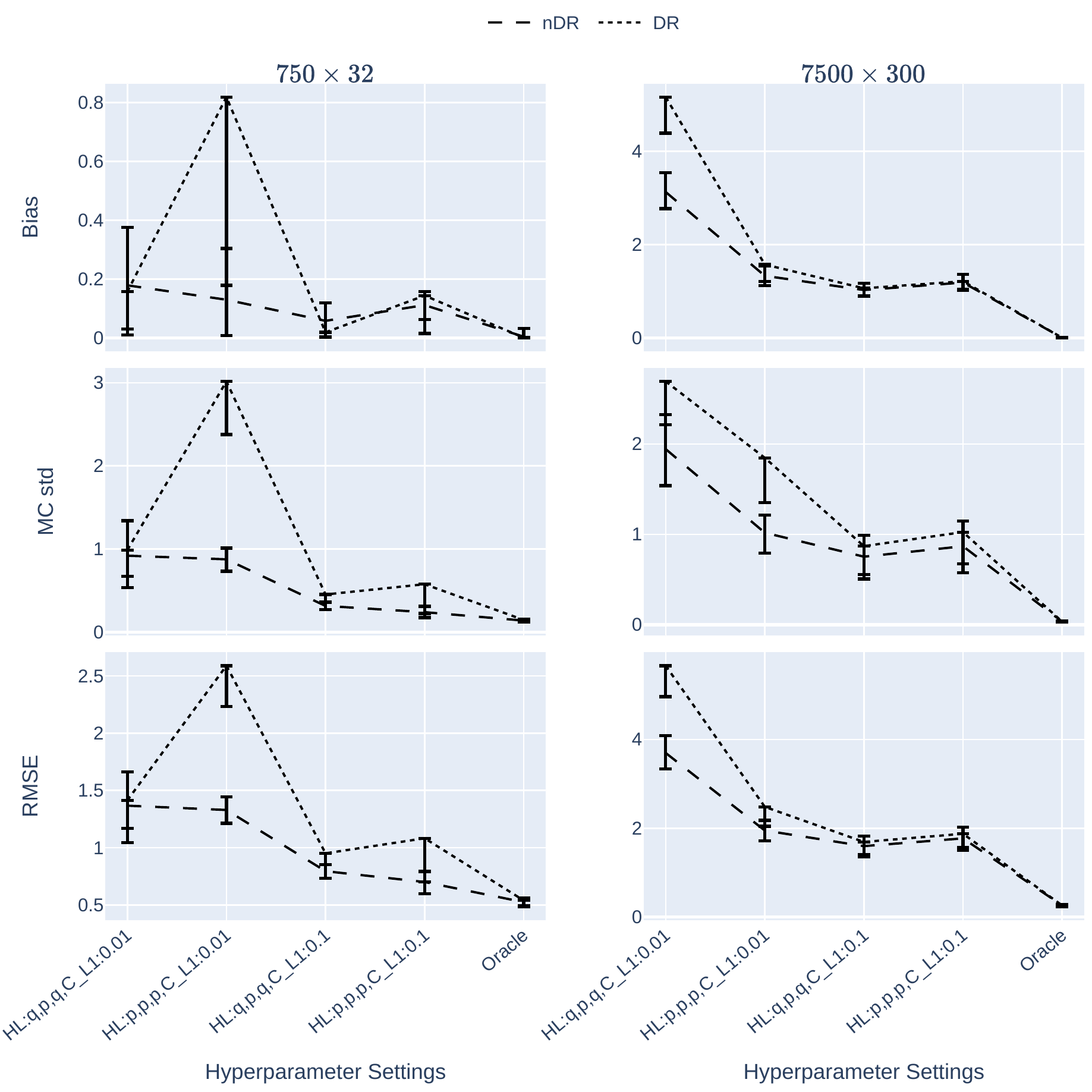}
\label{nAIPWvsAIPW}
\end{figure}


Figure \ref{nAIPWvsAIPW-se} illustrates how the theoretical standard error formulas perform in MC experiments, whether they are close to the MC standard deviations. In these two graphs, smaller does not necessarily imply superiority. In these graphs, closeness of MC std and SE for both AIPW and nAIPW is desired. In the left two scenarios where NN's complexity is high, the MC std and SE are far from each other. Also, in the hyperparameter scenarios where both width of the NNs is small and regularization is higher, the MC std and SE are well separated. The best scenario seems to be were only regularization if higher but there is enough number of neurons in all the 3 hidden layers.

\begin{figure}[ht!]
\centering
\caption{The MC standard deviation and the standard error of the AIPW and nAIPW estimators for different data sizes and NN hyperparameters ($L_1$ regularization and width of the network.)}
\includegraphics[scale=0.7]{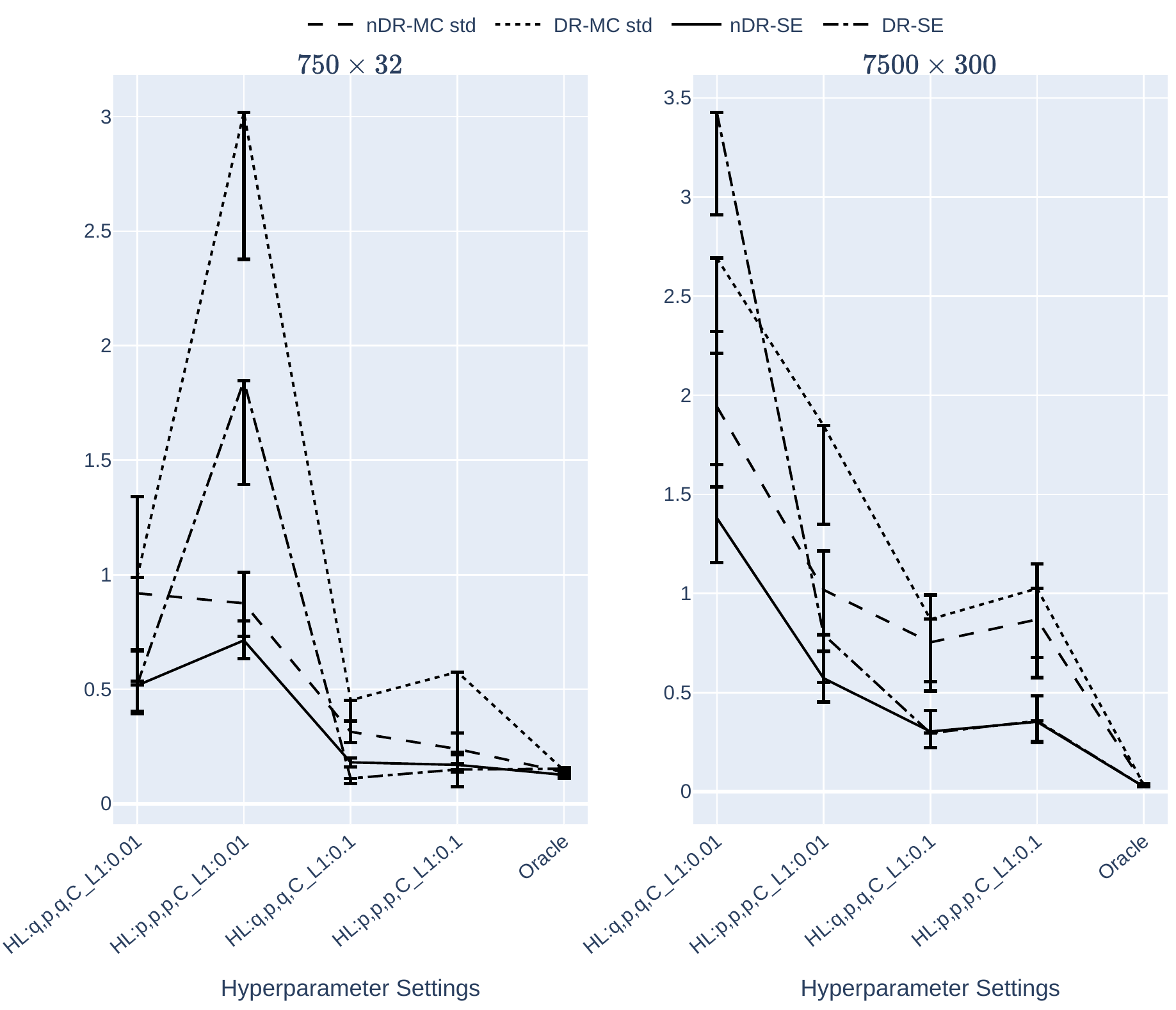}
\label{nAIPWvsAIPW-se}
\end{figure}


\section{Application: Food Insecurity and BMI}\label{applications}

The Canadian Community Health Survey (CCHS) is a cross-sectional survey that collects data related to health status, health care utilization and health determinants for the Canadian population in multiple cycles. The 2021 CCHS covers the population 12 years of age and over living in the ten provinces and the three territorial capitals. Excluded from the survey's coverage are: persons living on reserves and other Aboriginal settlements in the provinces; and some other sub-populations that altogether represent less than 3\% of the Canadian population aged 12 and over. Examples of modules asked in most cycles are: general health, chronic conditions, smoking, and alcohol use. For the 2021 cycle, thematic content on food security, home care, sedentary behaviour and depression, among many others, have been included. In addition to the health component of the survey are questions about respondent characteristics such as labour market activities, income, and socio-demographics.

In this article, we use CCHS dataset to investigate the Food Insecurity and Body Mass Index (BMI) causal relationship. Other gathered information in CCHS is used which might contain the potential confounders, y-predictors and instrumental variables. The data is a survey and needs special methods such as the resampling or bootstrap methods to estimate the standard errors. However, here, we use the data to illustrate the utilization of dNN on the causal parameters in case of positivity violation. In order to reduce the amount of variability in the data, we have focused on the subpopulation 18-65 years of age.

Figure \ref{cchs-figs} shows the ATE estimates and their 95\% asymptotic confidence intervals with nIPW, DR, and nDR methods, with 4 different neural networks which vary in terms of width and strength of $L_1$ reqularization. The scenario that results in the largest $R^2$ (as a measure of outcome prediction performance) outperforms the other scenarios. And the scenario that results in the largest AUC (as a measure of treatment model performance) results in the largest confidence intervals. This is because of more extreme propensity scores in this scenario. It is worth noting that the normalized IPW has smaller confidence intervals as compared to AIPW. However, as we do not know the truth about the ATE in this dataset, we can never know which estimator outperforms the other. To get an insight about this using the input matrix of this data, we simulated multiple treatments and outcomes with small to strong confounders and IVs and compared AIPW and nAIPW. In virtually all of them, the nAIPW is the favorite one. We are not presenting these results in this draft, but can be provided to the readers upon request.

\begin{figure}[ht!]
\centering
\caption{The ATE estimates and their asymptotically calculated 95\% confidence intervals with NIPW, DR, and nDR methods.}
\includegraphics[scale=0.75]{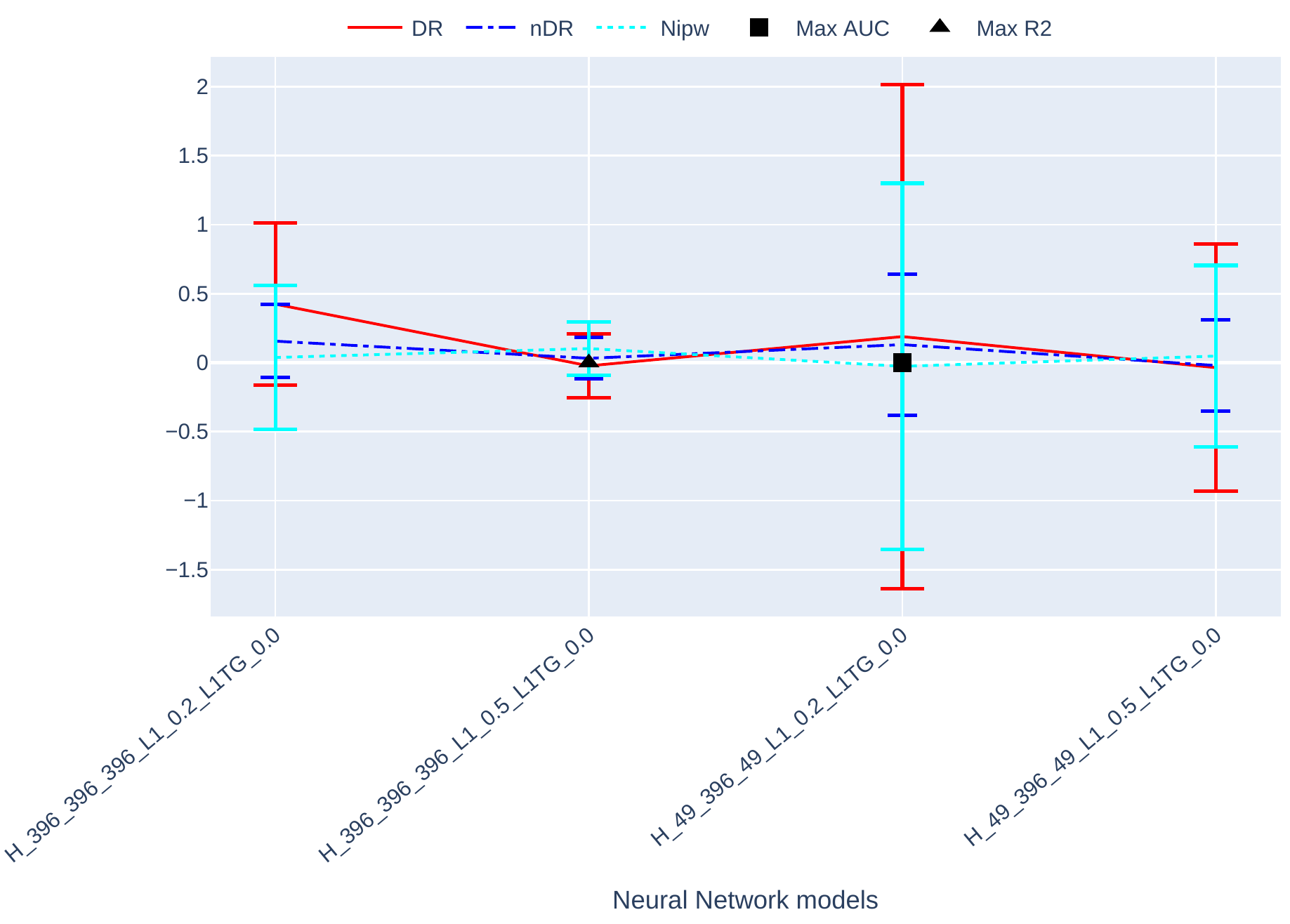}
\label{cchs-figs}
\end{figure}

\section{Discussion}\label{discussion}

Utilizing Machine Learning algorithms such as NNs in the first step estimation process is comforting as the concerns with regards to the non-linear relationships between the confounders and the treatment and outcome are addressed. However, there is no free lunch, and using NNs has its own caveats including theoretical as well as numerical challenges. \citet{farrell2018deep} addressed the theoretical concerns where they calculated the generalization bounds when two separate NNs are used to model the treatment and the outcome. However, they did not use or take into account regularization techniques such as $L_1$ or $L_2$ regularization. As NNs are complex algorithms, they provide perfect prediction for the treatment when the predictors are strong enough (or might overfit). Through monte carlo (MC) simulations, we illustrated that causal estimation and inference with double NNs can fail without the usage of regularization techniques such as $L_1$ and or extreme propensity scores are not taken care of. If $L_1$ regularization is not used, the normalization of the AIPW estimator (i.e. nAIPW) is advised to be employed as it dilutes the extreme predictions of the propensity score model and provide better bias, variance and RMSE. Our scenario analysis also showed that in case of violation of the positivity assumption in AIPW, normalization helps avoid blowing up the estimator (and standard error), but might be ineffective in taking into account confounding effects for some observations. An alternative might be trimming the propensity scores to avoid extreme values. However, the causal effect estimator will no longer be consistent and besides, there is no determined method where to trim. We hypothesize that
$\hat{h}^1=\hat{g}\hat{\mathbb{E}}\frac{A}{\hat{g}}\times I\big(\hat{g} \in (0, \epsilon)\big)+\hat{g}\times I\big(\hat{g} \in (\epsilon, 1)\big)$ and $\hat{h}^0=(1-\hat{g})\hat{\mathbb{E}}\frac{1-A}{1-\hat{g}}\times I\big(\hat{g} \in (1- \epsilon, 1)\big)+(1-\hat{g})\times I\big(\hat{g} \in (0, 1-\epsilon)\big)$
where $\epsilon = \frac{1}{n}$ will result in a consistent estimator, making the right assumptions and will outperform both of AIPW and nAIPW in case empirical positivity violation. We will study this hypothesis in a future article.

Another reason why NNs without regularization fail in the causal estimation and inference is that the networks are not targeted, and are not directly designed for these tasks. NNs are complex algorithms with strong predictive powers. This does not accurately server the purpose of causal parameter estimation, where the positivity assumption can be violated if strong confounders and/or instrumental variables \citep{angrist2008mostly} exist in the data. Ideally, the network should target the confounders and should be able to automatically limit the strength of predictors so that the propensity scores are not extremely close to 1 or 0. This was not investigated in this article and a solution to this problem is postponed to another research.

In Section \ref{applications} we applied the asymptotic standard errors of both AIPW and nAIPW, where the latter achieves smaller standard errors. That said, we acknowledge the fact that the asymptotic standard errors when using complex ML are not reliable, and in fact, they underestimate the calculated MC standard deviations, as illustrated in the simulations section \ref{simulations}. This is partly because of the usage of the complex algorithms such as NNs for estimation of the nuisance parameters in the first step. Further, the asymptotic distribution of the estimators are not symmetric (and thus are not normal). However, nAIPW is more symmetric than AIPW, according to the simulations while both estimators suffer from outliers. We will investigate the reasons and possible remedies for both asymptotic distribution and standard errors of the estimators in a future paper. 

\newpage
\section{Appendix}\label{appendix}

First let's review the proof sketch of the AIPW double robustness:

\eqref{trueAIPW} can be consistently estimated by
\begin{multline}\label{AIPWa}
    \hat{\beta}_{AIPW} = \frac{1}{n}\sum_{i=1}^n\Bigg[\Big(\frac{A_iY_i-\hat{Q}(1, W_i)(A_i-\hat{\mathbb{E}}[A_i |W_i])}{\hat{\mathbb{E}}[A_i |W_i]} \Big)
   -
   \Big(\frac{(1-A_i)Y_i+\hat{Q}(0, W_i)(A_i-\hat{g}_i)}{1-\hat{\mathbb{E}}[A_i |W_i]}\Big)\Bigg] = \\
   \frac{1}{n}\sum_{i=1}^n \Bigg(\big[\frac{A_i}{\hat{g}_i} - \frac{1-A_i}{1-\hat{g}_i}\big]y_i - \frac{A_i-\hat{g}_i}{\hat{g}_i(1-\hat{g}_i)}\big[(1-\hat{g}_i)\hat{Q}^1_i+\hat{g}_i\hat{Q}^0_i\big]\Bigg) = \\
   \frac{1}{n}\sum_{i=1}^n \Big(\frac{A_i(y_i-\hat{Q}^1_i)}{\hat{g}_i} - \frac{(1-A_i)(y_i-\hat{Q}^0_i)}{1-\hat{g}_i}\Big) + 
    \frac{1}{n}\sum_{i=1}^n \big(\hat{Q}^1_i-\hat{Q}^0_i\big)
\end{multline}

The second formula guarantees the consistency of AIPW if $\hat{g}$ is consistent, and the third expression the consistency of $\hat{Q}^0_i$ and $\hat{Q}^1_i$ are consistent.

\begin{theorem}[nAIPW Double Robustness]\label{AIPWproof}
Let nAIPW estimator of risk difference be 

\begin{equation}\label{empricnAIPW}
    \hat{\beta}_{nAIPW} = \hat{\mathbb{E}}(\hat{Q}^1-\hat{Q}^0) + \hat{\mathbb{E}}\Big(\frac{A(Y-\hat{Q}^1)}{\hat{g}\hat{\mathbb{E}}[\frac{A}{\hat{g}}]}-\frac{(1-A)(Y-\hat{Q}^0)}{(1-\hat{g})\hat{\mathbb{E}}[\frac{1-A}{1-\hat{g}}]}\Big).
\end{equation}
Then $\hat{\beta}_{nAIPW}$ is a consistent estimator of $\beta$ if $\hat{g}\xrightarrow{p}g$ or $\hat{Q}^k\xrightarrow{p}Q^k$, $k=0, 1$.
\end{theorem}

\begin{proof}

From \eqref{empricnAIPW}, $\hat{\beta}_{nAIPW}$ is a consistent estimator of $\beta$ if $\hat{Q}^0_i$ and $\hat{Q}^1_i$ are consistent. This is because the first term $\hat{\mathbb{E}}(\hat{Q}^1-\hat{Q}^0)$ converges to $\beta$, while the second term tends to zero. 

By re-expressing \eqref{empricnAIPW}, the other argument is clear. Let $\hat{w}^1=\hat{\mathbb{E}}[\frac{A}{\hat{g}}]$ and $\hat{w}^0=\hat{\mathbb{E}}[\frac{1-A}{1-\hat{g}}]$, we have:

\begin{multline}\label{empricnAIPW2}
    \hat{\beta}_{nAIPW} = \frac{1}{n}\sum_{i=1}^n \Bigg(\big[\frac{A_i}{\hat{g}_i\hat{w}^1_i} - \frac{1-A_i}{(1-\hat{g}_i)\hat{w}^0_i}\big]y_i\Bigg) + \hat{\mathbb{E}}\Big(\hat{Q}^1-\hat{Q}^0 - \frac{A_i\hat{Q}^1}{\hat{g}\hat{w}^1}+\frac{(1-A_i)\hat{Q}^0}{(1-\hat{g})\hat{w}^0}\Big)=\\
    \frac{1}{n}\sum_{i=1}^n \Bigg(\big[\frac{A_i}{\hat{g}_i\hat{w}^1_i} - \frac{1-A_i}{(1-\hat{g}_i)\hat{w}^0_i}\big]y_i -
    \hat{Q}^1_i\big(A_i-\hat{g}_i\hat{w}^1_i\big) + 
    \hat{Q}^0_i\big(1-A_i-(1-\hat{g}_i)\hat{w}^0_i\big)
    \Bigg)
\end{multline}

The first expression in \eqref{empricnAIPW2} is the same as nIPW estimator which is consistent estimator of $\beta$ \citep{lunceford2004stratification}. Now, under consistency of $\hat{g}$, the second term tends to zero , as $\hat{w}_1 \xrightarrow{p} 1$ and $\hat{w}_0 \xrightarrow{p} 1$.

In the theorem below it is shown that there is an M-estimation equivalent to $\beta_{nAIPW}$ and $w^1$ and $w^0$. This, plus the continuous mapping theorem proves that $\sum_{i=1}^n\frac{A_i}{\hat{g}_i}$ converges in probability to $n$ if $\hat{g}\xrightarrow{p}g$.



\end{proof}




\begingroup
\def\thetheorem{\ref{nAIPWvartheorem}}

\begin{theorem}

The asymptotic variance of the nAIPW \eqref{nAIPW} is

\begin{equation}\label{nAIPWvar1}
    \hat{\sigma}^2_{nAIPW} = \sum_{i=1}^n \Big(
    \frac{A_i(y_i-\hat{Q}^1_i)w_i^{(1)}}{\sum_{j=1}^n A_jw_j^{(1)}} -    \frac{(1-A_i)(y_i-\hat{Q}^0_i)w_i^{(0)}}{\sum_{j=1}^n (1-A_j)w_j^{(0)}} + \frac{1}{n}\big(\hat{\beta}_{SR} - \hat{\beta}_{nAIPW}\big) \Big)^2,
\end{equation}
where $\hat{Q}^k_i = \hat{Q}(k,W_i)$ and $\hat{g}_i = \hat{\mathbb{E}}[A_i |W_i]$. 

\end{theorem}

\addtocounter{theorem}{-1}
\endgroup

\begin{proof}

Let's define a few notations first:
\begin{equation} \label{notations}
\begin{split}
        q &= Q^1 - Q^0,\\
        g &= \mathbb{E}[A|W],\\
        f &= y - Q^1,\\
        h &= y - Q^0,\\
        v &= \frac{A}{g},\\
        u &= \frac{1-A}{1-g}.\\
\end{split}
\end{equation}

With this set of notations, the nAIPW estimator \eqref{nAIPW} can be written as

\begin{equation}\label{nAIPW11}
        \hat{\beta}_{nAIPW} = \sum_{i=1}^n \Big(\frac{v_if_i}{\sum_{j=1}^n v_j} - \frac{u_ih_i}{\sum_{j=1}^n u_j} + \frac{q_i}{n}\Big),
\end{equation}

Following the methods in \citep{stefanski2002calculus}, to find an estimating equation that whose solution is $\hat{\beta}_{nAIPW}$, we introduce two more estimating equations. Employing the M-Estimation theory, we will prove that nAIPW is asymptotically normal, and we will calculate its standard error.

It can be seen that \eqref{nAIPW11} is not a solution to an M-estimator directly. However, by defining two more parameters and concatenating their estimating equations we obtain a 3-dim multivariate estimating equations 
\begin{equation} \label{esteqs00}
\begin{split}
        \sum_{i=1}^n \Big(\frac{v_if_i}{\gamma} - \frac{u_ih_i}{\lambda} + \frac{1}{n}(q_i - \beta)\Big) &= 0,\\
        \sum_{i=1}^n \Big(v_i - \frac{\gamma}{n}\Big)&= 0,\\
        \sum_{i=1}^n \Big(u_i - \frac{\lambda}{n}\Big) &= 0.
\end{split}
\end{equation}

To ease the calculations, we modify the first estimating equation with an equivalent one, but the results will not differ:
\begin{equation} \label{esteqs}
\begin{split}
        \sum_{i=1}^n \lambda v_if_i - \gamma u_ih_i + \frac{\gamma \lambda}{n}(q_i - \beta) &= 0,\\
        \sum_{i=1}^n v_i - \frac{\gamma}{n}&= 0,\\
        \sum_{i=1}^n u_i - \frac{\lambda}{n} &= 0.
\end{split}
\end{equation}

By defining the following notations,
\begin{align*}\label{notations1}
\psi = \begin{pmatrix}
\phi\\
\eta\\
\Omega
\end{pmatrix}
 = \begin{pmatrix}
 \lambda vf - \gamma uh + \frac{\gamma \lambda}{n}(q - \beta)\\
 v - \frac{\gamma}{n}\\
 u - \frac{\lambda}{n}
\end{pmatrix},
\end{align*}
we have $\sum_{i=1}^n \psi_i = 0 $, or 
\begin{equation}\label{highdimee}
\begin{split}
        \sum_{i=1}^n \phi_i, &= 0,\\
        \sum_{i=1}^n \eta_i&= 0,\\
        \sum_{i=1}^n \Omega_i &= 0.
\end{split}
\end{equation}

The M-Estimation theory implies that under regulatory conditions, the solutions to these estimating equations converge in distribution to a multivariate normal distribution:
\begin{align*}
\sqrt{n}\begin{pmatrix}
\hat{\beta}_{nAIPW}\\
\hat{\gamma}\\
\hat{\lambda}
\end{pmatrix} &\sim  MVN
\begin{pmatrix}
\theta \ , \
\mathbf{I}^{-1}(\theta)\mathbf{B}(\theta)\mathbf{I}^{-1}(\theta)^T
\end{pmatrix}
\end{align*}
where 
\begin{align*}
\theta = \begin{pmatrix}
\beta\\
\gamma\\
\lambda
\end{pmatrix},
\end{align*}

\begin{align}\label{I}
\mathbf{I}(\theta) = -\mathbb{E}\frac{\partial \psi}{\partial \theta^T} = \frac{1}{n}\begin{pmatrix}
\frac{\lambda \gamma}{n} & \mathbb{E}(uh-\frac{\lambda}{n}(q-\beta)) & -\mathbb{E}(vf+\frac{\gamma}{n}(q-\beta))\\
0 & \frac{1}{n} & 0\\
0 & 0 & \frac{1}{n}
\end{pmatrix},
\end{align}
whose inverse is
\begin{align}\label{Iinv}
\mathbf{I}^{-1}(\theta) = \frac{n}{\gamma \lambda}\begin{pmatrix}
1 & -n\mathbb{E}(uh-\lambda(q-\beta)) & n\mathbb{E}(vf+\frac{\gamma}{n}(q-\beta))\\
0 & \gamma \lambda & 0\\
0 & 0 & \gamma \lambda
\end{pmatrix},
\end{align}
and,
\begin{align}\label{B}
\mathbf{B}(\theta) = \mathbb{E}\psi \psi^T = \begin{pmatrix}
\mathbb{E} \phi^2 & \mathbb{E} \phi\eta & \mathbb{E} \phi \Omega\\
\mathbb{E} \phi\eta& \mathbb{E} \eta^2 & \mathbb{E} \eta \Omega\\
\mathbb{E} \phi \Omega& \mathbb{E} \eta \Omega& \mathbb{E} \Omega^2
\end{pmatrix}.
\end{align}

In order to estimate the variance of $\hat{\beta}_{nAIPW}$, we do not need to calculate all entries of the variance-covariance matrix, only the first entry:

\begin{align}\label{IinvBInvT}
\frac{1}{n} (\frac{n^2}{(\gamma \lambda)^2}) \begin{pmatrix}
\mathbb{E} \phi^2 + \epsilon & \star & \star\\
\star& \star & \star\\
\star& \star& \star
\end{pmatrix}.
\end{align}

The $\star$ entries are irrelevant to the calculation of variance of nAIPW and the term $\epsilon$ is a very long expression which involves terms converge to zero faster than the actual estimating equations \eqref{highdimee} \citep{hines2021demystifying} (also verified by simulations):
\begin{multline}
    \epsilon = -\mathbb{E}\phi \eta (n\mathbb{E} uh+\lambda(\beta-q))+\mathbb{E}\phi \Omega(n\mathbb{E}vf-\gamma (\beta-q))-(n\mathbb{E} uh+\lambda(\beta-q))(-\mathbb{E}\eta^2(n\mathbb{E} uh+\lambda(\beta-q))+\\
    \mathbb{E}\eta \Omega(n\mathbb{E}vf-\gamma (\beta-q))+\mathbb{E}\phi \eta )+(n\mathbb{E}vf-\gamma (\beta-q))(-\mathbb{E}\eta \Omega(n\mathbb{E} uh+\lambda(\beta-q))+\mathbb{E}\Omega^2(n\mathbb{E}vf-\gamma (\beta-q))+\mathbb{E}\phi \Omega).
\end{multline}

Further, 

\begin{align}
\sqrt{n}\begin{pmatrix}
\hat{\beta}_{nAIPW}\\
\hat{\gamma}\\
\hat{\Omega}
\end{pmatrix} &\sim  MVN
\begin{pmatrix}
\theta \ , \
\hat{\mathbf{I}}^{-1}(\hat{\theta})\hat{\mathbf{B}}(\hat{\theta})\hat{\mathbf{I}}^{-1}(\hat{\theta})^T
\end{pmatrix}
\end{align}
where we replace $\mathbb{E}$ with sample averages in expression \eqref{I}-\eqref{B} and $\theta$ with their corresponding solutions to equations \eqref{esteqs}. Following this recipe, we get
\begin{equation}
    \hat{\sigma}^2_{nAIPW} = \frac{1}{n} (\frac{n^2}{(\gamma \lambda)^2}) \hat{\mathbb{E}} \phi^2 + \hat{\epsilon} \approx \\
    \sum_{i=1}^n\Big(\frac{v_if_i}{\hat{\gamma}} - \frac{u_ih_i}{\hat{\lambda}} + \frac{1}{n}{q_i} - \hat{\beta}_{nAIPW})\Big)^2,
\end{equation}
which is the same as \eqref{nAIPWvar1}.

\end{proof}

\bibliography{ref}

\end{document}